\documentclass[a4paper]{article}

\usepackage{amsmath}
\usepackage{amssymb}
\usepackage{amsthm}
\usepackage{amsfonts}
\usepackage{array}
\usepackage{bm}
\usepackage{indentfirst}
\usepackage{paralist}
\usepackage{geometry}
\usepackage{color}
\usepackage{multimedia}
\usepackage{multicol}
\usepackage{multirow}
\usepackage{bbm}
\usepackage{graphics}
\usepackage{graphicx}
\usepackage{booktabs}
\usepackage{subfigure}
\usepackage{grffile} 
\usepackage{float}
\usepackage{url}
\usepackage[labelsep=quad]{caption}
\usepackage{mathrsfs}
\usepackage{cases}
\usepackage{lscape}
\usepackage{ulem} 
\newcommand{\B}{\bm{\mathcal{B}}}

\newcommand{\R}{\bm{\mathbb{R}}}
\newcommand{\bxi}{\bm{\mathcal{\xi}}}
\newcommand{\btheta}{{\bm{\mathcal{\theta}}}}
\newcommand{\balpha}{{\bm{\mathcal{\alpha}}}}
\newcommand{\z}{\bm{z}}

\newcommand{\F}{\bm{F}}

\newcommand{\accept}[1]{\color{black}{#1}}

\newtheorem{remark}{Remark}
\newtheorem{proposition}{Proposition}
\newtheorem{assumption}{Assumption}
\newtheorem{lemma}{Lemma}
\newtheorem{theorem}{Theorem}[section]
\newtheorem{example}{Example}
\newtheorem{definition}{Definition}
\usepackage{tabularx}
\author{{Mengge Li} \footnote{{Department of Mathematics, NUS, li.mengge@u.nus.edu.} The work of the first author is partially supported by the Ministry of Education in Singapore under
		the grant MOE AcRF A-8000453-00-00.} \and Shuaijie Qian \footnote{CMSA, Harvard, shuaijie$\_$qian@fas.harvard.edu } \and Chao Zhou \footnote{{Department of Mathematics, NUS, matzc@nus.edu.sg.} The work of the third author is partially supported by the Ministry of Education in Singapore under
		the grant MOE AcRF A-8000453-00-00, A-0004273-00-00, A-0004589-00-00 and by NSFC under the grant award 11871364.}}
\title{Robust Equilibrium Strategy for Mean-Variance Portfolio Selection}

\begin{document}

	\maketitle
	\section*{Abstract}
	The classical mean-variance portfolio selection problem induces time-inconsistent {(\it{precommited})} strategies (see Zhou and Li (2000)). To {overcome} this time-inconsistency, Basak and Chabakauri (2010) introduce the game theoretical approach and look for {(sub-game perfect Nash)} \textit{equilibrium strategies}, {which is} solved from the corresponding partial differential equations (PDE) system. 
{In their model, the investor perfectly knows the drift and volatility of the assets.   
However, in reality investors only have an estimate on them, e.g, a 95\% confidence interval. In this case, some literature (e.g., Pham, Wei and Zhou (2022)) derives the optimal precommited strategy under the worst parameters, which is the \it{robust control}. }
{The relation between the equilibrium strategy and the PDE system has not been justified when incorporating robust control.
 } 
	In this paper, we consider a general dynamic mean-variance  framework and propose a novel definition of the \textit{robust equilibrium strategy}. Under our definition, a classical solution to the corresponding PDE system implies a robust equilibrium strategy. We then explicitly solve for some special examples.
	\section{Introduction}  
	
	Markowitz (1952) pioneers in the mean-variance portfolio selection problem, where the explicit solution to a single-period problem is proposed. \textcolor{black}{Later Zhou and Li (2000) investigate a dynamic continuous-time mean-variance portfolio selection problem in the spirit of Markowitz's  work.} 
However, there are two drawbacks in  this model. 

	First, the devised strategy is a \textit{precommited strategy}, that is, the investor finds the optimal strategy at initial time and sticks to it until maturity. However, this strategy may be suboptimal in the future which results in the time-inconsistency feature. A game theoretic approach to handle this problem is first proposed in Basak and Chabakauri (2010), where they derive a time-consistent explicit optimal strategy. However, in this strategy, the amount of wealth invested in stock is independent of the total wealth, which contradicts the common knowledge.  In view of this shortcoming, \textcolor{black}{Björk and Murgoci (2010) seek for the Nash sub-game prefect equilibrium for the time inconsistent problems with the Markov process. The corresponding strategy they {derive} is known as \textit{equilibrium startegy}. }Bj\"{o}rk, Murgoci, and Zhou (2014) extend the model to cases with the wealth-dependent risk aversion level. \textcolor{black}{ As an extension, Hernández and Possamaï (2021) develop the sub-game perfect Nash equilibrium strategy for non-Markovian time-inconsistent stochastic control {problems} and introduce the corresponding BSDE system}.  
	
	Second, Markowitz (1952) assumes that the investor perfectly knows the parameters of the market. However, this assumption is too strong for practical applications. Therefore, researchers incorporate the robustness into the model.  	
\textcolor{black}{In a static setting, Garlappi, Uppal and Wang (2007) take into account the drift uncertainty, while Liu and Zeng (2017) concern the uncertainty of the correlation matrix.} Pham, Wei, and Zhou (2022) \textcolor{black}{focus on the dynamic setting, and take into account} the uncertainty on the drift and correlation of multiple stocks. 
	
Zeng, Li, and Gu (2016), Pun (2018), and Yan et al. (2020) \textcolor{black}{investigate the dynamic equilibrium strategy with robustness concern}. However, Zeng, Li, and Gu (2016) ignore the theoretical foundation of this problem. More precisely, they do not show that the solution to the corresponding PDE must be {an equilibrium} strategy. \textcolor{black}{Pun (2018) and Yan et al. (2020) circumvent this obstacle via assuming that the ``nature" knows all possible decisions of the investor and select the worst-case scenario with time-consistent manner.}
		\begin{table}
		\begin{center}
			
			\begin{tabularx}{15cm}{X<{\centering} p{2cm}<{\centering} X<{\centering} X<{\centering}}
				\hline
			& Strategy  & Uncertainty & Issues\\ \hline
\hline
Pham, Wei and Zhou (2022)& Precommitted & Drift and correlation in product and ellipsoidal set& Time-inconsistency and no jump\\ \hline
Zeng, Li and Gu (2016) & Equilibrium & Drift and jumps intensity& Omit the proof of verification theorem \\\hline
Yan, Han, Pun, and Wong (2020)& Equilibrium & Drift& Unclear on why the worst-case scenario is time-consistent in the set of all equivalent probability measures\\\hline
Ours & Equilibrium &Drift, variance and jumps in product set&  $/$\\
				
				\hline
			\end{tabularx}
		\end{center}
	\end{table}

	In terms of how to model the uncertainty, there are two strands of literature. The first strand introduces the entropy penalty into utility maximization. \textcolor{black}{In this way, there is no constraint on the candidate market condition}.  Among this strand of literature, Maenhout (2004) concerns the robustness on  the stock return. Later Branger and Larsen (2013) take into account the uncertainty about jump and diffusion. Flor and Larsen (2014) focus on an investor uncertain about the drift of bonds and stocks. Jin, Luo and Zeng  (2020) investigate the uncertainty of jumps.  
	Another strand focuses on the portfolio selection under the worst scenario (see, e.g., Jin and Zhou (2015); Lin and Riedel(2014); Fouque, Pun, and Wong(2016)). \textcolor{black}{In this strand, the candidate market condition is constrained on a subset of all possible market conditions.} Among this strand, our work is most closely related to Pham, Wei, and Zhou (2022), but they seek a precommited strategy while we are interested in dynamic equilibrium strategies.

\subsection{Contribution}	
	In this paper, we  concern a general dynamic mean-variance problem. The generality comes from two parts. First, we cover general wealth processes. Second,  
	we take into account general mean-variance criterions, which covers the classical mean variance criteria for terminal wealth or for portfolio log returns (see Dai et al. (2020)), and the risk aversion coefficient can depend on wealth.  
	
		We propose a novel definition of robust equilibrium strategy.
		Basak and Chabakauri (2010) introduce the game theoretical approach which freezes the strategy in the future and looks for the sub-game perfect equilibrium. In this way, the derived strategy is time-consistent.  
	We extend their definition to incorporate robust control, and show that a classical solution to the corresponding PDE system still implies a robust equilibrium. 
	
	\textcolor{black}{Pun (2018) and Yan et al. (2020) also propose a definition of robust equilibrium strategy. 
	Essentially, they are looking for an equilibrium of a game between the environment and the investor. However, it is not natural to introduce the environment as a player, since the environment is exogenous. 
	In comparison, our definition emphasizes an investor looking for a sub-game perfect equilibrium under some worst condition s/he concerns.}
	In particular, the investor chooses the set of the possible models in a time-consistent manner, and s/he believes the ``nature" is the worst case from this set. An investor in a time-consistent manner means this investor only considers market conditions which does not depend on the starting time of the investment.
	\textcolor{black}{Moreover, we also extend to the case with jumps in stock price {dynamics}.} For some special models {arising} from portfolio selection problems, we can explicitly solve the corresponding PDEs, {and} find that the worst-case scenario is independent of \textcolor{black}{time, wealth level and risk aversion coefficients} in these models.
	
	 	
     
 The rest of the paper is organized as follows. Section 2 is devoted to the setup of a basic model to illustrate our novel definition of robust equilibrium strategies. Moreover, we also formally derive the corresponding PDE system.  
 Section 3 focuses on solving some special examples of the basic model in Section 2 to build some intuition on the optimal strategy. The implication for the optimal robust portfolio strategy is also presented there. In Section 4, we give a general model and rigorously show that a solution to the PDE system is a robust equilibrium strategy in our definition. Section 5 is the summary of this paper.  Some technical proofs and calculations are relegated to the appendix.  

\section{Basic Model {and Robust Equilibrium Strategy}}\label{sec model}
In this section, we assume the investor's self-financing wealth process satisfies the following general framework 
\begin{align}\label{equ self-fin process}
	d X^{\balpha, \btheta}_s = \eta(\balpha_s, \btheta_s) ds +  \bxi(\balpha_s, \btheta_s) \cdot d\B_s, \quad {t\leq s \leq T, \qquad X^{\balpha, \btheta}_t = x},
\end{align}
where {$\{\B_s\}_{t\leq s \leq T}$} is an $n$-dimensional standard Brownian motion {on a filtered probability space ($\Omega$, $\mathscr{F}$, $\{\mathscr{F}_s\}_{s \in[t,T ]}$, $P$)}.
We use $\btheta_s \in \Theta$ to represent the market {scenario}, which is {exogenous} and unknown to the investor. 
%
$\eta: \mathbb{R}^{n} \times  \Theta  \to \mathbb{R}$ and $\bxi:  \mathbb{R}^{n} \times  \Theta  \to \mathbb{R}^n$ are two functions representing the drift term and diffusion term, respectively. 
{$\balpha_s$ is the investor's control. The admissible control set is}
\begin{align}
	\mathscr{A}_t =& \{\balpha_s\in \R^{n}, t{\leq} s \leq T \big| \balpha_s\ \text{is adapted to}\ \mathscr{F}_s, \notag \\
	& \qquad \text{and}\ E\big[\int_t^T \|\bxi(\balpha_s, \btheta_s)\|_2^2ds \big]< +\infty, {E\big[\int_t^T |\eta(\balpha_s, \btheta_s)|ds \big]< +\infty}, \quad {\forall  \{\btheta_s\}_{t\leq s \leq T}  \in \Theta_{[t, T]}}  \}. \notag
\end{align}
where 
\begin{align}
\Theta_{[t, T]}: = \{\btheta_s, t\leq s \leq T| \btheta_s \in \Theta\}. \notag
\end{align} 
For initial state $X^{\balpha, \btheta}_t = x \in \mathbb{R}$, we define a functional, {which is related to the target of the investor}, of the {following} form 
	\begin{align}\label{equ func J}
		J(t, x; \balpha, \btheta) = {E}_t[X^{\balpha, \btheta}_{T}]-\lambda Var_t(X^{\balpha, \btheta}_{T}) = {E}_t[X^{\balpha, \btheta}_{T}-\lambda (X^{\balpha, \btheta}_{T})^2] +  \lambda {E}_t^2[X^{\balpha, \btheta}_T]. 
	\end{align} 
	\begin{example}\label{exam wealth}
		A special case {related to} our general model {\eqref{equ self-fin process}-\eqref{equ func J}} is the mean-variance problem for terminal wealth {(see, e.g., Markowitz (1952) and Bjork, Murgoci, and Zhou (2014))}, where $X_t$ is the wealth process and $\balpha$ is the amount of money in {stocks}. In  this setting, there are $n$ stocks and the stock prices $\bm{\mathcal{S}}_t : = (\bm{\mathcal{S}}_{1t},\ldots,\bm{\mathcal{S}}_{nt})^\top$ evolve as 
		\begin{align}\label{equ stock price}
			\frac{d \bm{\mathcal{S}}_{it}}{\bm{\mathcal{S}}_{it}} = \bm{b}_{it} dt + \sum_{j=1}^n \bm{\sigma}_{ijt} d \B_{jt},
		\end{align}
	    {and the scenario is 
	    \begin{align}\label{equ ambiuity}	
	    	\btheta_t = (\bm{b}_t, \bm{\Sigma}_t),
    	\end{align}
   	 	where $\bm{b}_t=(\bm{b}_{1t}, \bm{b}_{2t}, ..., \bm{b}_{nt})^\top$, $\bm{\Sigma}_t = \bm{\sigma}^\top_t\bm{\sigma}_t$, } $\bm{\sigma}_t = (\bm{\sigma}_{ijt})_{n\times n}$ \textcolor{black}{and $\B_{t}=(\B_{1t}, \B_{2t},..., \B_{nt})^{\top}$}. The wealth process is
		\begin{align}
			d X_t^{\balpha, \btheta} = \balpha_t^\top \bm{b}_t dt  + \balpha_t^\top \bm{\sigma}_t d \B_{t}.\notag
		\end{align}
		In this special case, $\eta(\balpha_t, \btheta_t) = \balpha_t^\top \bm{b}_t$ and $\bm{\xi}(\balpha_t, \btheta_t) = (\balpha_t^\top \bm{\sigma}_t)^\top$.
	\end{example}
	
	\begin{example}\label{exam log wealth}
		Dai et al. (2020) consider the mean-variance problem for log returns. In this case, the stock price {and the scenario} still follows  \eqref{equ stock price} {and \eqref{equ ambiuity}, respectively}. By denoting $\balpha_t$ as the proportion of wealth in stock,  the concerned $\log$ wealth process satisfies
		\begin{align}
			d X_t^{\balpha, \btheta} = (\balpha_t^\top \bm{b}_t - \frac{1}{2} \balpha_t^\top \bm{\Sigma}_t \balpha_t) dt  + \balpha_t^\top \bm{\sigma}_t d \B_t.\notag
		\end{align}
		Then ${\eta}(\balpha_t, \btheta_t) =\balpha_t^\top \bm{b}_t - \frac{1}{2} \balpha_t^\top \bm{\Sigma}_t \balpha_t$ and $\bm{\xi}(\balpha_t, \btheta_t) = (\balpha_t^\top \bm{\sigma}_t)^\top$, where $\bm{\Sigma}_t :=\bm{\sigma}_t (\bm{\sigma}_t)^\top$.
	\end{example}
	
	Following Bjork, Murgoci and Zhou (2014), we only consider feedback controls, i.e., the control $\alpha$ is a deterministic function of the state variable $(t, x)$. To make our notations simple, we still denote this function as $\balpha(t, x)$. {Similarly, we take $\btheta$ as a deterministic function $\btheta(t, x)$, too.}     
	
{In the follows, we define the robust equilibrium strategy.}	
\begin{definition}((Time-consistent) worst-case scenario)
	
	For a given strategy $\balpha(s, y)$, $(s, y)\in [t, T]\times \mathbb{R}$,  we say $\btheta^{\balpha}(s, y)$ to be a worst-case scenario {(for $\balpha$)} if 
	\begin{align}
		\limsup \limits_{h \to 0} \frac{J(t, x; \balpha, \btheta^{\balpha}) - J(t, x; \balpha, \btheta^{\balpha}_{h, \bm{u}})}{h} \leq 0, \notag
	\end{align} 
	for any $(t, x)$ and $\btheta^{\balpha}_{h, \bm{u}}$ defined as 
	\begin{numcases}{\btheta^\balpha_{h, \bm{u}}(s, y) = }
		\bm{u} & $t\leq s < t+h$\notag\\
		\btheta^\balpha(s, y) & $t+h\leq s \leq T$,\label{equ defi theta}
	\end{numcases}  
	where $\bm{u}\in \Theta$ is a constant. 
\end{definition}
\color{black}{In our framework, the investor {chooses market condition} 
in a time-consistent way. Therefore, when given a strategy $\alpha$, the investor looks for the worst case $\btheta^\alpha$, such that any local variation of $\btheta^\alpha$ {at time $t$} will make $\balpha$ better.}
{In the follows, we always use $\btheta^{\balpha}$ to denote the worst-case scenario for $\alpha$.}	

\begin{definition}(Equilibrium)

	Given a strategy $\balpha(s,y)$ with $(s,y) \in [t,T] \times \mathbb{R}$, we construct a strategy ${\balpha_{h, w}}$ by 
	\begin{numcases}{{\balpha_{h, w}}(s, y) = }
		\bm{w}, & $t\leq s < t+h$\notag\\
		\balpha(s, y), & $t+h\leq s \leq T$,\notag
	\end{numcases}  
	where $\bm{w}\in \R^{n}$ is a constant, $h>0$, and $(t, x)$ is arbitrarily chosen. 
	We say  $\balpha$ is a robust equilibrium strategy if 
	\begin{align}
		\liminf \limits_{h\to 0} \frac{ J(t, x; \balpha, \btheta^\balpha) -  \inf \limits_{\bm{u}\in \Theta}J(t, x; {\balpha_{h, w}}, \btheta^{\balpha_{h,w}}_{h,\bm{u}})}{h} \geq 0	\notag
	\end{align}
	for any $\bm{w}$ and $(t, x)$, {where $\btheta^{\balpha_{h,w}}_{h,\bm{u}}$ is the same as \eqref{equ defi theta} }.  
\end{definition}	
\color{black}{In this definition, the investor optimizes $\balpha$ such that it performs better than any of its local variation under the respective worst-case scenarios.

\color{black}{Our definition of robust equilibrium strategy is inspired of Basak and Chabakauri (2010) and Bj\"{o}rk, Murgoci, and Zhou (2014), in which they assume the market parameters are perfectly known by the investor. Their problems can be regarded as a special case of ours by setting $\btheta$ as a set with only one element. Later Pun (2018) introduces model uncertainty into this problem, and propose a definition of robust equilibrium, which emphasizes the equilibrium between the nature and investor. In comparison, our definition here focuses on the investor's own target of finding a sub-game perfect equilibrium in a game-theoretical setting, since we believe the nature has no intention to play against an investor.}

\subsection{PDE Approach}
	Denote $\hat{\balpha}$ as a robust equilibrium strategy, and \textcolor{black}{assume} ${\btheta^{\hat{\balpha}}}$ is the corresponding worst market condition.
	Then $V(t, x) := J(t, x; \hat{\balpha}, {\btheta^{\hat{\balpha}}})$ 
	and $g(t, x) = {E}[X^{\hat{\balpha}, {\btheta^{\hat{\balpha}}}}_{T}|X^{\hat{\balpha}, {\btheta^{\hat{\balpha}}}}_{t}=x]$ formally satisfy the following PDE system.
	\begin{equation}\label{equ ori PDE}
		\left\{\begin{aligned}
			& \sup_{\balpha\in {\R^{n}}}\inf_{\btheta {\in \Theta}}\{\mathcal{L}^{\balpha, \btheta}V(t,x)-\lambda \mathcal{H}^{\balpha, \btheta} g(t, x)\} \\
			&\qquad \qquad = \mathcal{L}^{\hat{\balpha}, {\btheta^{\hat{\balpha}}}}V(t,x)-\lambda \mathcal{H}^{\hat{\balpha}, {\btheta^{\hat{\balpha}}}} g(t, x)=0\\
			&\mathcal{L}^{\hat{\balpha}, {\btheta^{\hat{\balpha}}}}g(t,x)=0\\
			&V(T,x)=x\\
			&g(T,x)=x,
		\end{aligned}\right.
	\end{equation}
	where 
	\begin{align}\label{equ opera A}
		&\mathcal{L}^{\balpha,\btheta}\psi(t,x)=\psi_{t}(t,x)+\eta(\balpha, \btheta) \psi_{x}(t, x)+\frac{1}{2}\|\bm{\xi}\|_2^2(\balpha, \btheta) \psi_{xx}(t,x)\\
		& \mathcal{H}^{\balpha,\btheta}\psi(t,x) = \|\bm{\xi}\|_2^2(\balpha, \btheta) \psi_{x}^{2}(t, x).
	\end{align}
	Following Bjork, Murgoci and Zhou (2014), we can also set $$f(t,x)={E}_t[X_{T}^{\hat{\balpha}, {\btheta^{\hat{\balpha}}}}-\lambda (X_{T}^{\hat{\balpha}, {\btheta^{\hat{\balpha}}}})^{2}|X^{\hat{\balpha},{\btheta^{\hat{\balpha}}}}_{t}=x],$$ then we have 
	\begin{align}
		V(t,x)=f(t,x)+\lambda g^{2}(t,x). \notag
	\end{align}
	The corresponding PDE system is
	\begin{equation}\label{equ2 ori PDE}
		\left\{\begin{aligned}
			&\sup_{\balpha\in \R^{n}}\inf_{\btheta {\in \Theta}}\{ \mathcal{L}^{\balpha,\btheta}f(t,x)+2\lambda g(t,x) \mathcal{L}^{\balpha,\btheta}g(t,x)\}\\
			&\qquad \qquad = \mathcal{L}^{\hat{\balpha},{\btheta^{\hat{\balpha}}}}f(t,x)+2\lambda g(t,x) \mathcal{L}^{\hat{\balpha},{\btheta^{\hat{\balpha}}}}g(t,x) = 0\\
			&\mathcal{L}^{\hat{\balpha}, {\btheta^{\hat{\balpha}}}}f(t,x)=0\\
			&\mathcal{L}^{\hat{\balpha}, {\btheta^{\hat{\balpha}}}}g(t,x)=0\\
			&f(T,x)=x-\lambda x^{2}\\
			&g(T,x)=x.
		\end{aligned} \right.
	\end{equation}
	
	\begin{proposition}
		We have the following equivalence between the two PDE systems \eqref{equ ori PDE} and \eqref{equ2 ori PDE}.
		
		1. Assume {$(V, g)$} is a pair of solutions to the  PDE system \eqref{equ ori PDE}, and $V$ and $g$ are $C^{1,2}$ in $(t, x)$, then {$(f, g)$} is a pair of solutions to the  PDE system \eqref{equ2 ori PDE} with 
		$$f(t, x): = V(t, x) - \lambda g^2(t, x).$$
		
		2. Assume {$(f, g)$} is a pair of solutions to the  PDE system \eqref{equ2 ori PDE}, and $f$ and $g$ are $C^{1,2}$ in $(t, x)$, then {$(V, g)$} is a pair of solutions to the  PDE system \eqref{equ ori PDE} with 
		$$V(t, x) := f(t, x)+ \lambda g^2(t, x).$$
	\end{proposition}
{The proof of this proposition is straightforward, we omit here.}
	
	\section{Explicit solution to Special Examples}\label{sec expli solu}

{In this section, we solve the PDE systems {associated with}
	the two examples in  Section \ref{sec model}.   
{All calculations are relegated to Appendix.}
We inherit the assumption ($H\Theta$) from Pham, Wei and Zhou (2022) on the set $\Theta$. For readers' convenience, we {give} the assumption below.  }
\begin{assumption}\label{assump zhou mf}

{We assume the set $\Theta= \{(\bm{b}, \bm{\Sigma}) \}$ is of the following form}:\par
$\Theta = \Delta \times \Gamma$, where $\Delta$ is a compact set in $\mathbb{R}^n$, and $\Gamma$ is a convex subset of $\mathbb{S}_{>+}^{n}$, the set of all positive definite symmetric $n\times n$ matrices. \footnote{\color{black}{Pham, Wei and Zhou (2022) assume the {marginal} volatilities are perfectly known, while the correlation matrix is uncertain. Apart from this ``product set" case, they also introduce the ``ellipsoidal set" case, which can also be handled by this paper's approach  if {marginal} volatilities are known.} }  \\	
\end{assumption}
\subsection{Special Examples}
{As illustrated in Pham, Wei and Zhou (2022), Assumption \ref{assump zhou mf} implies ambiguity in drift is independent of the ambiguity of the \textcolor{black}{correlation} matrix.
	
	For the mean-variance problem of terminal wealth in Example  \ref{exam wealth}, we have the following solution {to} PDE system \eqref{equ ori PDE},
	\begin{numcases}{}
		V(t,x)=x+\frac{1}{4\lambda}\hat{\bm{b}}^{\top}\hat{\bm{\Sigma}}^{{-1}}\hat{\bm{b}}(T-t)\notag\\
		g(t,x)=x+\frac{1}{2\lambda}\hat{\bm{b}}^{{\top}}\hat{\bm{\Sigma}}^{{-1}}\hat{\bm{b}}(T-t)\notag\\
		f(t,x)=(1-\hat{\bm{b}}^{{\top}}\hat{\bm{\Sigma}}^{{-1}}\hat{\bm{b}}(T-t))(x+\frac{1}{4\lambda}\hat{\bm{b}}^{{\top}}\hat{\bm{\Sigma}}^{{-1}}\hat{\bm{b}}(T-t))-\lambda x^{2}\notag\\
		\hat{\balpha}=\frac{1}{2\lambda}\hat{\bm{\Sigma}}^{{-1}}\hat{\bm{b}},\notag
	\end{numcases}
	where $\btheta^{\hat{\balpha}}=(\hat{\bm{b}},\hat{\bm{\Sigma}})=\arg \min \limits_{\btheta \in \Theta} \bm{b}^{\top} \bm{\Sigma}^{-1} \bm{b}$.
	
	For the mean-variance problem of terminal $\log$ return in Example  \ref{exam log wealth}, we have the following solution {to} PDE system \eqref{equ ori PDE}
	\begin{numcases}{}
		V(t,x)=x+\frac{1}{2(1+2\lambda)}\hat{\bm{b}}^{{\top}}\hat{\bm{\Sigma}}^{{-1}}\hat{\bm{b}}(T-t)\notag\\
		g(t,x)=x+\frac{2\lambda+\frac{1}{2}}{(1+2\lambda)^{2}}\hat{\bm{b}}^{{\top}}\hat{\bm{\Sigma}}^{{-1}}\hat{\bm{b}}(T-t)\notag\\
		f(t,x)=(1-\frac{4\lambda^{2}+\lambda}{(1+2\lambda)^{2}}\hat{\bm{b}}^{{\top}}\hat{\bm{\Sigma}}^{{-1}}\hat{\bm{b}}(T-t))x-\lambda x^{2}\notag\\
		\qquad +\frac{1}{1+2\lambda}\hat{\bm{b}}^{{\top}}\hat{\bm{\Sigma}}^{{-1}}\hat{\bm{b}}(T-t)[\frac{1}{2}-\frac{\lambda (2\lambda+\frac{1}{2})^{2}}{(1+2\lambda)^{3}}\hat{\bm{b}}^{{\top}}\hat{\bm{\Sigma}}^{{-1}}\hat{\bm{b}}(T-t)]\notag\\
		\hat{\balpha}=\frac{1}{1+2\lambda}\hat{\bm{\Sigma}}^{{-1}}\hat{\bm{b}}\notag
	\end{numcases}
	where $\btheta^{\hat{\balpha}}=(\hat{\bm{b}},\hat{\bm{\Sigma}})=\arg \min \limits_{\btheta\in \Theta} \bm{b}^{\top}{\bm{\Sigma}}^{-1}{\bm{b}}$.

	\subsection{Worst case scenario}
	Let $\Theta=\{(b,\Sigma)\}$ be a product set with the product constraint:
	\begin{align}
		&b_1\in [\underline{b}_1,\overline{b}_1],\quad b_2\in [\underline{b}_2,\overline{b}_2],\quad \sigma_1\in [\underline{\sigma}_1,\overline{\sigma}_1],\quad\sigma_2 \in [\underline{\sigma}_2,\overline{\sigma}_2],\quad -1<\underline{\rho}\leq \rho\leq \overline{\rho}<1.
	\end{align}	
	Notice that in this case, $\Sigma=\begin{pmatrix}
		\sigma^{2}_{1} & \sigma_{1}\sigma_{2}\rho\\ 
		\sigma_{1}\sigma_{2}\rho & \sigma^{2}_{2}
	\end{pmatrix}$, which is positive definite symmetric obviously. {We also} have $\Gamma$ is a convex set.\\
	In the product constraint case, the worst case is selected via the following criterion: 
	\begin{align}\label{equ V}
		(\hat{b}_1,\hat{b}_2,\hat{\sigma}_1,\hat{\sigma}_2,\hat{\rho}):=\arg \min \limits_{b_1, b_2, \sigma_1,\sigma_2,\rho} \bm{b}^{\top} \bm{\Sigma}^{-1} \bm{b} = \frac{1}{\sigma_1^{2}\sigma_2^{2}(1-\rho^2)}\bigg(\sigma_2^{2}b_1^2+\sigma_1^{2}b_2^2 -2 \sigma_1\sigma_2\rho b_1 b_2\bigg).
	\end{align}  
	We consider the case where $\underline{b}_i \geq 0$ and $\underline{\sigma}_i > 0$, $i = 1,2$. \accept{We assume 
\begin{align}	\label{assump b sigma}
\overline{b}_2\leq \overline{b}_1,  \underline{b}_2\leq \underline{b}_1, \overline{\sigma}_2\geq \overline{\sigma}_1,\ \text{and}\ \underline{\sigma}_2\geq \underline{\sigma}_{1}.
\end{align}
}
\begin{proposition}
	 Let the optimal portfolio strategy be $\mathbf{\hat{\alpha}}=\binom{\hat{\alpha}_1} {\hat{\alpha}_2}$. Then we have the following possible cases:\\ 
		(1) If $\underline{\rho}>\frac{\overline{b}_2/\underline{\sigma}_2}{\underline{b}_1/\overline{\sigma}_1}$, the worst case scenario is $\hat{\rho}=\underline{\rho}, \hat{b}_1=\underline{b}_1,\hat{b}_2=\overline{b}_2,\hat{\sigma}_1=\overline{\sigma}_1,\hat{\sigma}_2=\underline{\sigma}_2$, the optimal portfolio is $\hat{\alpha}_1=\frac{1}{2\lambda(1-\underline{\rho}^{2})}(\underline{b}_1/\overline{\sigma}_1-\underline{\rho} \overline{b}_2/\underline{\sigma}_2)>0$, and $\hat{\alpha}_2=\frac{1}{2\lambda(1-\underline{\rho}^{2})}(\overline{b}_2/\underline{\sigma}_2-\underline{\rho} \underline{b}_1/\overline{\sigma}_1)<0$.\\
		(2)If $\overline{\rho}<\frac{\underline{b}_2/\overline{\sigma}_2}{\underline{b}_1/\overline{\sigma}_1}$, the worst case scenario is $\hat{\rho}=\overline{\rho}, \hat{b}_1=\underline{b}_1,\hat{b}_2=\underline{b}_2,\hat{\sigma}_1=\overline{\sigma}_1,\hat{\sigma}_2=\overline{\sigma}_2$, the optimal portfolio is $\hat{\alpha}_1=\frac{1}{2\lambda(1-\overline{\rho}^{2})}(\underline{b}_1/\overline{\sigma}_1-\overline{\rho} \underline{b}_2/\overline{\sigma}_2)>0$, and $\hat{\alpha}_2=\frac{1}{2\lambda(1-\overline{\rho}^{2})}(\underline{b}_2/\overline{\sigma}_2-\overline{\rho} \underline{b}_1/\overline{\sigma}_1)>0$.\\ 
		(3) Otherwise, the worst case scenario is $\hat{b}_1=\underline{b}_1, \hat{\sigma}_1=\overline{\sigma}_1$, \accept{and $\hat{b}_2, \hat{\sigma}_2, \hat{\rho}$ is chosen to satisfy $\frac{\hat{b}_2}{\hat{\sigma}_2} = \hat{\rho} \frac{\underline{b}_1}{\bar{\sigma}_1}$ }, the optimal portfolio is $\hat{\alpha}_1=\frac{1}{2\lambda}\underline{b}_1/\overline{\sigma}_1$, and $\hat{\alpha}_2=0$. 
	\end{proposition}
\begin{proof}
	See proof of Proposition 2 in Appendix.
\end{proof}
\accept{In the above Case 1, $\frac{b_2}{\sigma_2}$ takes the maximal value allowed, which is counterintuitive,  since the ``worst case" usually corresponds to maximal volatility and minimal return. Our explanation is as follows. Mathematically, the worst case is to find the worst $\bm{b}^{\top} \bm{\Sigma}^{-1} \bm{b}$.  Noticing that the right side of \eqref{equ V} can be written as
\begin{align}\label{equ V new}
	 \frac{1}{(1-\rho^2)}\bigg(\frac{b_1^2}{\sigma_1^2}+\frac{b_2^2}{\sigma_2^2} -2 \rho \frac{b_1}{\sigma_1} \frac{b_2}{\sigma_2}\bigg),
\end{align}  
and the derivative on $\frac{b_2}{\sigma_2}$ is always negative for a high $\rho$. 
 Financially, when both $\frac{b_1}{\sigma_1}$ and $\rho$ are high, the investor takes a short position on stock 2, and thus, a higher $\frac{b_2}{\sigma_2}$ implies a lower value of the criterion. 
	  }
	\begin{remark}(Financial Interpretation)
		Drift $b$ and marginal volitality $\sigma$ take effect of the worst case scenario in the form of Sharpe ratio $b/\sigma$. In the first case, when {correlation coefficient} is large enough, the investor {can benefit from hedging}, thus he  longs one stock which has larger $b/\sigma$ and shorts the other. In the second case, when the correlation coefficient is small, {the hedging benefit can not cover the loss from shorting an asset with positive return rate. }the investor longs both stocks. In the last case, investor will only long the stock with higher Sharpe ratio and ignore the other.
	\end{remark}
		
	\section{Model Extension with Jumps and State-Dependent Risk Aversion}
	We can further introduce jump ambiguity and state-dependent risk aversion level into our framework {in Section \ref{sec model}}. More precisely, we {assume the self-financing wealth process is} {a L\'evy process as follows:}
	\begin{align}\label{equ wealth gene}
		X^{\balpha, \btheta}_s = & \int_t^s \eta(\balpha_v, \btheta_v) dv + \int_t^s \bxi(\balpha_v, \btheta_v) d\B_v \notag \\
		&+ \int_{v\in [t, s], \|\z\|_2\geq1} {\zeta(X^{\balpha, \btheta}_{v-}, \balpha_v, \z)} N(dv, d\z) \notag \\
		&+  \int_{v\in [t, s], \|\z\|_2< 1} \zeta(X^{\balpha, \btheta}_{v-}, \balpha_v, \z) \bigg(N(dv, d\z) - \F^{\btheta_v}(d\z) dv\bigg),\quad {s\in [t, T]},\quad {X^{\balpha, \btheta}_t = x,}
	\end{align}
	where  $\z\in \mathbb{R}^k\backslash\{0\}$, {$\zeta: \mathbb{R} \times \mathbb{R}^n\times \mathbb{R}^k \to \mathbb{R}$}.
	{The jump measure $N$ is a Poisson random measure on $[0, +\infty)\times \R^k$ with intensity measure 
	$\F^{\btheta_v}$}, {which depends on the scenario $\btheta_v$. Moreover, for any $\btheta \in \Theta$, $\F^\btheta$ is a Radon measure on $\R^k\backslash\{0\}$ with 
	\begin{align}
		\int_{\|\z\|_2 <1} \|\z\|_2^2 \F^{\btheta}(d\z) < +\infty, \quad  	\int_{\|\z\|_2 \geq 1} \F^{\btheta}(d\z) < +\infty. \notag
	\end{align}
}
{The admissible control set becomes
\begin{align}
	\mathscr{A}_t =& \bigg\{\balpha_s\in \R^{n}, t{\leq} s \leq T \big| 
	 E\big[\int_t^T \|\bxi(\balpha_s, \btheta_s)\|_2^2ds \big]< +\infty, {E\big[\int_t^T |\eta(\balpha_s, \btheta_s)|ds \big]< +\infty},\notag \\
&\qquad \zeta(X^{\balpha, \btheta}_{s-}, \balpha_s, \z)\in \R, \quad E\Bigg[\int_{s\in [t, T], \|\z\|_2\geq1} |\zeta(X^{\balpha, \btheta}_{s-}, \balpha_s, \z)| \F^{\btheta_s}(d\z) ds\Bigg]< +\infty,\notag \\
	 &\qquad E\Bigg[ \int_{s\in [t, T], \|\z\|_2< 1} \zeta^2(X^{\balpha, \btheta}_{s-}, \balpha_s, \z)  \F^{\btheta_s}(d\z) ds\Bigg]<+\infty,
	  \quad \forall 
	  {\btheta_s \in \Theta, t\leq s \leq T,}\ \z \in \R^k\backslash\{0\}  \bigg\}. \notag
\end{align}
}
	
{Since we introduce wealth-dependent risk aversion level,} the functional \eqref{equ func J} {changes into}
	\begin{align}
		J(t, x; \balpha, \btheta) = {E}_t[X^{\balpha, \btheta}_{T}]-\lambda(x) Var_t(X^{\balpha, \btheta}_{T}) = {E}_t[X^{\balpha, \btheta}_{T}-\lambda(x) (X^{\balpha, \btheta}_{T})^2] +  \lambda(x) {E}_t^2[X^{\balpha, \btheta}_T], \notag
	\end{align}  
	where $\lambda(x)$ is a deterministic function of $x$. Denote 
\begin{numcases}{}	
	V(t, x) =\sup_{\balpha\in\mathscr{A}_t}\inf_{\btheta \in \Theta_{(t, T]} } J(t, x; \balpha, \btheta),\notag\\ 
	g(t, x) =  E[X_{T}^{\hat{\balpha}, {\btheta^{\hat{\balpha}}}}], \notag\\ 
	f^y(t,x)= f(t, x, y),\notag
\end {numcases} 
where $$f(t, x, y) = E[X_{T}^{\hat{\balpha},{\btheta^{\hat{\balpha}}}}-\lambda(y)(X_{T}^{\hat{\balpha}, {\btheta^{\hat{\balpha}}}})^{2}],$$ 
{which separates the wealth level $x$ from the risk aversion parameter $\lambda(y)$.} 
Then $V(t,x)=f^x(t,x)+\lambda(x)g^{2}(t,x)$. The corresponding PDE system is
	\begin{align}\label{equ2 exten PDE}
		\left\{\begin{matrix}
			\sup\limits_{\balpha\in \R^{n}}\inf\limits_{\btheta\in \Theta}\{ \mathcal{A}^{\balpha,\btheta}f^{x}(t,x)+2\lambda(x)g(t, x) \mathcal{A}^{\balpha,\btheta}g(t,x)\}=0\\
			\mathcal{A}^{\hat{\balpha}, {\btheta^{\hat{\balpha}}}}f^x(t,x)=0\\
			\mathcal{A}^{\hat{\balpha}, {\btheta^{\hat{\balpha}}}}g(t,x)=0\\
			f^x(T,x)=x-\lambda(x)x^{2}\\
			g(T,x)=x,
		\end{matrix}\right.
	\end{align}	
	where 
	\begin{align}\label{equ opera A}
		&\mathcal{A}^{\balpha,\btheta}\psi(t,x)=\psi_t+ \eta(\balpha, \btheta) \psi_{x}(t, x)+\frac{1}{2}\|\bxi\|_2^2(\balpha, \btheta) \psi_{xx}(t,x)\notag \\
		&+\int \bigg(\psi(t,x+\zeta(x,\balpha,\z))-\psi(t,x)- \zeta(x,\balpha, \z)\mathbf{1}_{\|\z\|_2<1} \psi_{x}(t, x)\bigg)\F^{\btheta}(d\z),
	\end{align}
	and {$\hat{\balpha}$, ${\btheta^{\hat{\balpha}}}$} realize the {$\sup\limits_{\balpha\in \R^{n}}\inf\limits_{\btheta\in \Theta}$} in the first equation of \eqref{equ2 exten PDE}. 
	
	In the following, we will verify that a solution to the PDE system is related to an equilibrium.  Before introducing the theorem, we first list the assumptions we need. 
	\begin{assumption} \label{assump 2}
		1. For any $(x, \balpha) \in \mathbb{R}\times {\R^{n}}$, $\zeta(x, \balpha, \z)$ is locally uniformly bounded w.r.t. $\z$. Moreover, there is a constant $K>0$, such that $\limsup \limits_{\|\z\|_2\to \infty} \frac{\zeta(x, \balpha, \z)}{\|\z\|_2^K} = 0$. \\ 
		2. $\int_{\|\z\|_2<1} \|\z\|_2^2 F^\btheta(d\z) + \int_{\|\z\|_2\geq1} \|\z\|_2^{K^2} F^\btheta(d\z)$ is bounded for any  $\btheta \in \Theta$. 
	\end{assumption}
	
	\begin{theorem}\label{thm veri}
	\accept{Under Assumption \ref{assump 2},} assume we find a smooth solution $(f, g)$ to the  PDE system \eqref{equ2 exten PDE}, and 
		\begin{align}\label{equ condi V veri}
			&\sup \limits_{t_0\leq t\leq T}\limsup \limits_{|x|\to +\infty} \frac{{f^x(t, x)+\lambda(x)g^{2}(t,x) }}{|x|^K} = 0,\  \text{for any $t_0 <T$}.
		\end{align}
		Then {the implied strategy} $\hat{\balpha}$ is a robust equilibrium {strategy}, {the implied $\btheta^{\hat{\balpha}}$ is the corresponding worst-case scenario,} and 
		\begin{align}
			V(t, x): = f^x(t,x)+\lambda(x)g^{2}(t,x) \notag
		\end{align}
		is the  value function.  
	\end{theorem} 
\begin{proof}[Proof of Theorem \ref{thm veri}]
The following lemma plays an important role in proving the main Theorem.
\begin{lemma}\label{lem main thm}
	For any constant $\balpha$ and $\btheta$, the $C^{1, 2}$ function $v(t, x)$ {satisfying \eqref{equ condi V veri}}, {we have}
	\begin{align}
		\liminf \limits_{h\to 0} \mathcal{A}^{\balpha, \btheta}v(t+h, x) = \mathcal{A}^{\balpha, \btheta}v(t, x).
	\end{align} 
\end{lemma}
\begin{proof}
	
	Since $v$ is $C^{1,2}$, we have 
	\begin{align}
		&\lim \limits_{h\to 0} v_t(t+h,x)+ \eta(\balpha, \btheta) v_{x}(t+h, x)+\frac{1}{2}\|\bxi\|_2^2(\balpha, \btheta) v_{xx}(t+h, x) \notag \\
		= &  v_t(t, x)+ \eta(\balpha, \btheta) v_{x}(t, x)+\frac{1}{2}\|\bxi\|_2^2(\balpha, \btheta) v_{xx}(t, x).
	\end{align}
	For the jump term in $\mathcal{A}^{\balpha,\btheta}$,  
	\begin{align}
		&\int \bigg(v(t+h, x+\zeta(x,\balpha,\bm{z}))-v(t+h, x)- \zeta(x,\balpha, \bm{z})\mathbf{1}_{||\bm{z}||_2<1} v_{x}(t+h, x)\bigg)\F^{\btheta}(d\bm{z}) \notag\\
		=& \int \bigg[\bigg(v(t+h, x+\zeta(x,\balpha,\bm{z}))-v(t+h, x)- \zeta(x,\balpha, \bm{z})\mathbf{1}_{||\bm{z}||_2<1} v_{x}(t+h, x ) \bigg) \notag \\
		&\qquad - \bigg(v(t, x+\zeta(x,\balpha,z))-v(t, x)- \zeta(x,\balpha, \bm{z})\mathbf{1}_{||\bm{z}||_2<1} v_{x}(t, x ) \bigg)\bigg]\F^{\btheta}(d\bm{z})  \notag \\
		&\quad +   \int\bigg(v(t, x+\zeta(x,\balpha,\bm{z}))-v(t, x)- \zeta(x,\balpha, \bm{z})\mathbf{1}_{||\bm{z}||_2<1} v_{x}(t, x ) \bigg)\F^{\btheta}(d\bm{z}). \notag
	\end{align}
	We only need to show 
	\begin{align}\label{equ conver inte term}
		&  \int \bigg[\bigg(v(t+h, x+\zeta(x,\balpha,\bm{z}))-v(t+h, x)- \zeta(x,\balpha, \bm{z})\mathbf{1}_{||\bm{z}||_2<1} v_{x}(t+h, x ) \bigg) \notag \\
		&\qquad - \bigg(v(t, x+\zeta(x,\balpha,\bm{z}))-v(t, x)- \zeta(x,\balpha, \bm{z})\mathbf{1}_{||\bm{z}||_2<1} v_{x}(t, x ) \bigg)\bigg]\F^{\btheta}(d\bm{z})   \to 0. 
	\end{align}
	Since $v$ is $C^{1,2}$, for any $M >0$, $v(t+h, \cdot)$ ($v_x (t+h, \cdot)$) converges uniformly on any compact set $K_{\bm{z}}$ to $v(t, \cdot)$ ($v_x (t, \cdot)$) as $h \to 0$. Then for any compact set $K_{\bm{z}}$
	\begin{align}
		&\lim \limits_{h\to 0}\int_{K_{\bm{z}}} \bigg[\bigg(v(t+h, x+\zeta(x,\balpha,\bm{z}))-v(t+h, x)- \zeta(x,\balpha, \bm{z})\mathbf{1}_{||\bm{z}||_2<1} v_{x}(t+h, x ) \bigg) \notag \\
		&\qquad - \bigg(v(t, x+\zeta(x,\balpha,\bm{z}))-v(t, x)- \zeta(x,\balpha, \bm{z})\mathbf{1}_{||\bm{z}||_2<1} v_{x}(t, x ) \bigg)\bigg]\F^{\btheta}(d\bm{z})  = 0,
	\end{align}  
	On the complementary set $K_{\bm{z}}^c$, we have 
	\begin{align}
		&|\int_{K_{\bm{z}}^c} \bigg[\bigg(v(t+h, x+\zeta(x,\balpha,\bm{z}))-v(t+h, x)- \zeta(x,\balpha, \bm{z})\mathbf{1}_{||\bm{z}||_2<1} v_{x}(t+h, x ) \bigg) \bigg] \F^{\btheta}(d\bm{z}) | \notag \\
		&\quad + |\int_{K_{\bm{z}}^c} \bigg[\bigg(v(t, x+\zeta(x,\balpha,\bm{z}))-v(t, x)- \zeta(x,\balpha, \bm{z})\mathbf{1}_{||\bm{z}||_2<1} v_{x}(t, x ) \bigg) \bigg] \F^{\btheta}(d\bm{z}) | \notag \\
		\leq & \int_{K_{\bm{z}}^c}  \bigg[C+ o(\bm{z}^{K^2})\bigg] \F^{\btheta}(d\bm{z})\label{equ esti compa remain}
	\end{align}  
	for some constant $C>0$. 
	For any $\epsilon>0$, we can choose a compact set $K^\epsilon_{\bm{z}}$, such that for any $h>0$, $\eqref{equ esti compa remain} \leq \epsilon$.
	That proves  \eqref{equ conver inte term}. 
\end{proof}

Next, we prove the main theorem.
Fix $(t, x)$, consider strategy $\hat{\balpha}$ and $\hat{\balpha}_h$, then by definition,  for any $s\geq t+h$, $y\in \mathbb{R}$ and $u\in\mathbb{R}$,
\begin{align}
	& J(s, y; \hat{\balpha}_h,  {\btheta}^{\hat{\alpha}}_{h, u}) = J(s, y; \hat{\balpha},  \btheta^{\hat{\balpha}}) = V(s, y), \notag \\
	&{E}_{s, y}[X_{T}^{\hat{\balpha}_h, \btheta^{\hat{\balpha}}_{h, u}}]  = {E}_{s, y}[X_{T}^{\hat{\balpha},\btheta^{\hat{\balpha}}}] = g(s, y),\notag \\
	&E_{s, y}[X_{T}^{\hat{\balpha}_h, \btheta^{\hat{\balpha}}_{h, u}}]-\lambda(z)E_{s, y}[(X_{T}^{\hat{\balpha}_h, \btheta^{\hat{\balpha}}_{h, u}})^{2}] 
	=  E_{s, y}[X_{T}^{\hat{\balpha},\btheta^{\hat{\balpha}}}]-\lambda(z)E_{s, y}[(X_{T}^{\hat{\balpha}, \btheta^{\hat{\balpha}}})^{2}]  = f(s, y, z), \notag
\end{align} 
{where $\btheta^{\hat{\balpha}}_{h, u}$ is defined as  \eqref{equ defi theta} }.

By definition of $\hat{\balpha}$, and $\btheta^{\hat{\balpha}}$, for any $\epsilon >0$, we can find $u_\epsilon \in \Theta$, such that  
\begin{align}
	{	\mathcal{A}^{\hat{\balpha}_h, \btheta^{\hat{\balpha}}_{h, u_\epsilon}}f^x(t, x) +2\lambda(x) g(t, x)\mathcal{A}^{\hat{\balpha}_h,  \btheta^{\hat{\balpha}}_{h, u_\epsilon}}g(t, x) 
		\leq  \mathcal{A}^{\hat{\balpha}, \btheta^{\hat{\balpha}}}f^x(t, x) +2\lambda(x) g(t, x)\mathcal{A}^{\hat{\balpha}, \btheta^{\hat{\balpha}}}g(t, x)+ \epsilon. \notag}
\end{align}
According to Lemma \ref{lem main thm}, we have  
\begin{align}
	&J(t, x; \hat{\balpha}, \btheta^{\hat{\balpha}}) - \inf \limits_{u} J(t, x; \hat{\balpha}_h,  \btheta^{\hat{\balpha}}_{h, u}) \notag \\
	=& \big(J(t, x; \hat{\balpha}, \btheta^{\hat{\balpha}}) -V(t, x) \big) -  \inf \limits_{u} \big(J(t, x; \hat{\balpha}_h,  \btheta^{\hat{\balpha}}_{h, u})  -V(t, x) \big) \notag \\
	\geq & \big(J(t, x; \hat{\balpha}, \btheta^{\hat{\balpha}}) -V(t, x) \big) -   \big(J(t, x; \hat{\balpha}_h,  \btheta^{\hat{\balpha}}_{h, u_\epsilon})  -V(t, x) \big) \notag \\
	=& \left[\bigg(\mathcal{A}^{\hat{\balpha}, \btheta^{\hat{\balpha}}}f^x(t, x) +2\lambda(x){g(t, x)} \mathcal{A}^{\hat{\balpha}, \btheta^{\hat{\balpha}}}g(t, x)\bigg) h + o(h)\right]\notag \\
	& \quad - \left[\bigg(\mathcal{A}^{\hat{\balpha}_h,  \btheta^{\hat{\balpha}}_{h, u_\epsilon}}f^x(t, x) +2\lambda(x)  {g(t, x)} \mathcal{A}^{\hat{\balpha}_h,  \btheta^{\hat{\balpha}}_{h, u_\epsilon}}g(t, x)\bigg) h + o(h)\right]\notag \\
	= & \left[\bigg(\mathcal{A}^{\hat{\balpha}, \btheta^{\hat{\balpha}}}f^x(t, x) +2\lambda(x) {g(t, x)} \mathcal{A}^{\hat{\balpha}, \btheta^{\hat{\balpha}}}g(t, x)\bigg) -  \bigg(\mathcal{A}^{\hat{\balpha}_h,  \btheta^{\hat{\balpha}}_{h, u_\epsilon}}f^x(t, x) +2\lambda(x)  {g(t, x)}\mathcal{A}^{\hat{\balpha}_h,  \btheta^{\hat{\balpha}}_{h, u_\epsilon}}g(t, x)\bigg) \right]h+o(h).\notag
\end{align}
Then we have
\begin{align}
	&\liminf \limits_{h\to 0} \frac{J(t, x; \hat{\balpha}, \btheta^{\hat{\balpha}}) - \inf \limits_{u} J(t, x; \hat{\balpha}_h,  \btheta^{\hat{\balpha}}_{h, u})}{h}\notag \\
	\geq &\bigg(\mathcal{A}^{\hat{\balpha}, \btheta^{\hat{\balpha}}}f^x(t, x) +2\lambda(x) {g(t, x)} \mathcal{A}^{\hat{\balpha}, \btheta^{\hat{\balpha}}}g(t, x)\bigg)\notag\\
	&\quad \qquad  - \bigg(\mathcal{A}^{\hat{\balpha}_h, u_\epsilon}f^x(t, x) +2\lambda(x)  {g(t, x)}\mathcal{A}^{\hat{\balpha}_h,  u_\epsilon}g(t, x)\bigg) \notag \\
	\geq &-\epsilon. \notag
\end{align}
That completes our proof by noticing the arbitrariness of $\epsilon$.
\end{proof}

\subsection{Explicit Solution for Some Particular Cases}\label{sec expli exten}
In this subsection, we give the explicit solution to some special cases for the mean-variance problem of terminal wealth. {All calculations are relegated to Appendix.} 


\begin{example}
We first introduce {Poisson jumps} into Example \ref{exam wealth}. More precisely, the stock dynamic is
\begin{align}
\frac{d \bm{\mathcal{S}}_{it}}{\bm{\mathcal{S}}_{it}} = \bm{b}_{it} dt + \sum_{j=1}^n \bm{\sigma}_{ijt} d \B_{jt} + \sum_{l=1}^{k}\bm{J}_{il} {Y}_{lt} d\bm{N}_{lt}, \quad i = 1,\ldots,n, \label{equ com poisson}
\end{align} 
{where} ${Y}_{lt}$ is the jump of type l with probability density $\bm{\Phi}_{l}(t,dy)$, {$l = 1,2,.., k$}, and $\bm{N_t}=(N_{1t},...,N_{kt})$ is a standard k-dimensional multivariate Poisson process with intensity $\bm{\mu}_t:k\times 1$, $\bm{J}:n\times k$ is the jump coefficient matrix with jump scaling coefficient $\bm{J}_{il}\in [0,1]$ for each $i,l$, which are all constants. Without loss of generality, we assume $rank(\bm{J})=k$. $\bm{\B_{t}},\bm{N_{t}},\bm{Y_{t}}$ are mutually independent, where $\bm{Y_{t}}: k\times k$ is a diagonal matrix with diagonal entries $Y_{1t},...,Y_{kt}$. 
{The scenario is defined as 
\begin{align}
\btheta_t = (\bm{b}_t, \bm{\Sigma}_t, \bm{\mu}_t, \bm{J}, \bm{\Phi} ).
\end{align}
}
Then the corresponding wealth follows \eqref{equ wealth gene} with 
\begin{align}
\eta(\balpha, \btheta) = \balpha^\top {\tilde{\bm{b}} },
\qquad \bm{\xi}(\balpha, \btheta) = (\balpha^\top \bm{\sigma})^\top, \qquad {\zeta(x, \balpha, \z) = \balpha^\top\bm{J}}\bm{Y}, \notag
\end{align}
 where $\tilde{\bm{b}}_i := \bm{b}_i +  \sum_{l=1}^{k}\bm{\mu}_{l}\mathbb{E}[\bm{Y}_{l}]\bm{J}_{il}\mathbf{1}_{\{\sum_{l=1}^{k}|\bm{\mu}_{l}\mathbb{E}[\bm{Y}_{l}]\bm{J}_{il}|<1\}}$, $i = 1,2,..., n$.

Notice that we need to keep our portfolio strategy $\alpha$ admissible, i.e., the wealth process {$X^{\alpha}_{t} \geq 0$} almost surely. Since $Y_{l,t}$ is mixed jumps, that is $Y_{l,t} \in (-1,\infty)$, according to {Jin, Luo, and Zeng (2021)}, $\alpha$ must satisfy the non-bankruptcy condition $\alpha^{\top}J_{l} \in [0,1]$ for each $l$.\\
\accept{Denote ${\bm{b}}_{\bm{F}}=\bm{b}+\bm{J}E[\bm{Y}]\bm{\mu}$ and $\bm{\Sigma_{F}}=\bm{\Sigma}+\sum_{l=1}^{k}\bm{\mu}_{l}\bm{J_{l}}\bm{J_{l}^{\top}}{E}[\bm{Y}_{l}^{2}]$ with $\bm{J}_{l}$ is the $l$-th column of matrix $\bm{J}$, and assume $({\bm{b}}_{\bm{F}},\bm{\Sigma_{F}} )$ satisfies Assumption \ref{assump zhou mf},
then we have the following solution. }
\begin{numcases}{}
	V(t,x)=x+\frac{1}{4\lambda}\hat{\bm{b}}_{\bm{F}}^\top\hat{\bm{\Sigma}}_{\bm{F}}^{-1}\hat{\bm{b}}_{\bm{F}}(T-t)\notag\\
	f^x(t,x)=-\lambda x^{2}+\bigg(1-\hat{\bm{b}}_{\bm{F}}^\top\hat{\bm{\Sigma}}_{\bm{F}}^{-1}\hat{\bm{b}}_{\bm{F}}(T-t)\bigg)x \notag \\
	\qquad \qquad +\frac{1}{4\lambda}\hat{\bm{b}}_{\bm{F}}^\top\hat{\bm{\Sigma}}_{\bm{F}}^{-1}\hat{\bm{b}}_{\bm{F}}(T-t)\bigg(1-\hat{\bm{b}}_{\bm{F}}^\top\hat{\bm{\Sigma}}_{\bm{F}}^{-1}\hat{\bm{b}}_{\bm{F}}(T-t)\bigg)\notag\\
	g(t,x)=x+\frac{1}{2\lambda}\hat{\bm{b}}_{\bm{F}}^\top\hat{\bm{\Sigma}}_{\bm{F}}^{-1}\hat{\bm{b}}_{\bm{F}}(T-t)\notag\\
	\hat{\balpha}=\frac{1}{2\lambda}\hat{\bm{\Sigma}}_{\bm{F}}^{-1}\hat{\bm{b}}_{\bm{F}}, \qquad 
\btheta^{\hat{\balpha}} = \arg \min\limits_{\btheta\in \Theta} {\bm{b}}_{\bm{F}}^\top{\bm{\Sigma}}_{\bm{F}}^{-1}{\bm{b}}_{\bm{F}}.    \notag  
\end{numcases}
\end{example}

\begin{example}
We next consider the case where $\lambda(x) = \frac{\lambda}{x}$ with jump under L\'evy measure $\bm{F}$. More precisely, the stock dynamic is
\begin{align}
	\frac{d \bm{\mathcal{S}}_{it}}{\bm{\mathcal{S}}_{it}} = \bm{b}_i dt + \sum_{j=1}^n \bm{\sigma}_{ij} d \B_{jt} +  \int_{\z} \bm{z} N(d\z, dt), \quad i = 1,\ldots,n, \notag
\end{align} 
{with scenario
\begin{align}
	\btheta_t = (\bm{b}_t, \bm{\Sigma}_t, \bm{F}_t).
\end{align}	
}		 
The corresponding wealth follows \eqref{equ wealth gene} with 
\begin{align}
	\eta(\balpha, \btheta) = \balpha^\top \bigg(\bm{b}+ \int_{\|\z\|_2<1} \z \bm{F}(d\z)\bigg) , \qquad \bm{\xi}(\balpha, \btheta) = (\balpha^\top \bm{\sigma})^\top, \qquad \zeta(x, \balpha, \z) = \balpha^\top \z. \notag
\end{align}
Denote \textcolor{black}{for any $\theta$,} \accept{${{\bm{b}}}_{\bm{F}}:={\bm{b}}+\int_{\bm{z}} \z {\bm{F}}(d\z)$ and ${\bm{\Sigma}}_{\bm{F}}:={\bm{\Sigma}}+M_{\bm{F}}$ with }
\begin{align*}
	M_{\bm{F}}:&=\left[
	\begin{array}{cccc}
		\int z_1^2\bm{F}(d\z) & 
		\int z_1z_2\bm{F}(d\z) & \cdots & 
		\int z_1z_n\bm{F}(d\z) \\
		\int z_1z_2\bm{F}(d\z) & 
		\int z_2^2\bm{F}(d\z) & \cdots &  
		\int z_2z_n\bm{F}(d\z)\\
		\vdots & \vdots & \ddots & \vdots \\
		\int z_1z_n \bm{F}(d\z) & 
		\int z_2z_n \bm{F}(d\z)& \cdots & 
		\int z_n^2\bm{F}(d\z) \\
	\end{array}
	\right].
\end{align*}
{If $({\bm{b}}_{\bm{F}},\bm{\Sigma_{F}} )$ satisfies Assumption \ref{assump zhou mf}}, we have the following solution: 
\begin{numcases}{}
	V(t,x)=\{A(t)+\lambda [A^2(t)-B(t)]\}x\notag\\
	f^x(t,x)=A(t)x-{\lambda}B(t)x\notag\\
	g(t,x)=A(t)x\notag\\
	\hat{\balpha}=\frac{1}{2\lambda}\hat{\bm{\Sigma}}_{F}^{-1}\hat{\bm{b}}_{F}\frac{A(t)+2\lambda [A^{2}(t)-B(t)]}{B(t)}x\notag\\
	\btheta^{\hat{\balpha}}=\arg \inf_{\btheta \in \Theta} \bm{b}_{\bm{F}}^{\top}\bm{\Sigma}_{\bm{F}}^{-1}\bm{b}_{\bm{F}},\notag
\end{numcases}
where $A$ and $B$ solves the ODE system
\begin{numcases}{}
	A_{t}+(\frac{A+2\lambda[A^{2}-B]}{2\lambda B}\hat{\bm{b}}_{\bm{F}}^{\top}\hat{\bm{\Sigma}}_{\bm{F}}^{-1}\hat{\bm{b}}_{\bm{F}})A=0\notag\\
	A(T)=1\notag\\
	B_{t}+\{2(\frac{A+2\lambda[A^{2}-B]}{2\lambda B}\hat{\bm{b}}_{\bm{F}}^{\top}\hat{\bm{\Sigma}}_{\bm{F}}^{-1}\hat{\bm{b}}_{\bm{F}})+\frac{(A+2\lambda[A^{2}-B])^{2}}{4\lambda^{2} B^{2}}\hat{\bm{b}}_{\bm{F}}^{\top}\hat{\bm{\Sigma}}_{\bm{F}}^{-1}\hat{\bm{b}}_{\bm{F}}\}B=0\notag\\
	B(T)=1. \notag
\end{numcases}
\end{example}

In the above examples, the worst-case scenario is constant. {Especially}, this worst-case scenario is independent of the risk aversion function $\lambda(x)$. 
\textcolor{black}{\begin{example}
		We can also introduce composited Poisson process \eqref{equ com poisson} into Example 2. More precisely,  the log wealth dynamic is 
		\begin{align}
d \ln W_t^{\balpha, \btheta} = (\balpha_t^\top \bm{b}_t - \frac{1}{2} \balpha_t^\top \bm{\Sigma}_t \balpha_t) dt  + \balpha_t^\top \bm{\sigma}_t d \B_t+\ln(\mathbf{e}_n+\bm{\alpha}_{t}^{\top}\bm{J}\bm{Y}_{t})d\bm{N}_{t},\notag
\end{align}
where $\bm{e_n}$ is the $n\times 1$ vector with all elements equals 1, and $\ln(\mathbf{a})= (\ln a_1, \ln a_2, ..., \ln a_n)^\top$ for any $n\times 1$ positive vector $\mathbf{a}$. 
Then
\begin{align}
	\eta(\balpha, \btheta) = \balpha^\top \bm{b}-\frac{1}{2}\balpha^{\top}\bm{\Sigma}\balpha ,
	\qquad \bm{\xi}(\balpha, \btheta) = (\balpha^\top \bm{\sigma})^\top, \qquad \zeta(x, \balpha, \z) =  \ln(\mathbf{e}_n+\bm{\alpha}^{\top}\bm{J}\bm{Y}_{t}), \notag
\end{align}
Since there is no close form solution, we leave the model here only.
\end{example}}

\section{Conclusion}
In this paper,  we propose a new definition of robust equilibrium strategy for a dynamic mean-variance problem.  
Our model is general enough to incorporate portfolio selection problem with wealth jumps, wealth-dependent risk aversion coefficient and mean-variance criterion for terminal portfolio wealth or log returns.
Compared with previous literature, our definition is more natural and intrinsic, {which is shown from the fact that the corresponding PDE system solution implies a robust equilibrium strategy.}
 We then explicitly solve some specific models and find that the worst-case scenario is independent of  time, wealth, and mean-variance criterion.  We also find the drift and marginal volitality take effect on the worst case scenario in the form of Sharpe ratio.


\section{Acknowledgments}
Mengge Li was partially supported by the Ministry of Education in Singapore under the grant MOE AcRF A-8000453-00-00. Chao Zhou was partially supported by the Ministry of Education in Singapore under the grant MOE AcRF A-8000453-00-00, A-0004273-00-00, A-0004589-00-00 and by NSFC under the grant award 11871364.

\bibliographystyle{plain}

\clearpage

\section{Appendix}
\subsection{Explicit solutions in Section \ref{sec expli solu}}
\subsubsection{Mean-variance w.r.t. wealth process}\label{sec basic wealth}
First we fix the market condition $\btheta= (\bm{b}, \bm{\Sigma})$ to be constant,	then it degenerates into the classical mean-variance problem. The corresponding PDE system is
	\begin{numcases}{}
		\sup_{\balpha\in \R^{n}}\{\mathcal{L}^{\balpha,\btheta} \bar{V}(t,x)-\lambda \mathcal{H}^{{\balpha},\btheta}   \bar{g}(t,x) \}= \mathcal{L}^{\hat{\balpha},\btheta} \bar{V}(t,x)-\lambda \mathcal{H}^{\hat{\balpha},\btheta}   \bar{g}(t,x) =0\notag\\
		\mathcal{L}^{ \hat{\balpha},\btheta} \bar{g}(t,x)=0 \notag\\
		 \bar{V}(T,x)=x \notag\\
		 \bar{g}(T,x)=x,\notag
	\end{numcases}
where {$\eta(\balpha, \btheta) = \balpha^\top \bm{b}$ and $\bm{\xi}(\balpha, \btheta) = (\balpha^\top \bm{\Sigma})^\top$}.
Then we have the following solution:
\begin{numcases}{}
\bar{V}(t,x)=x+\frac{1}{4\lambda} \bm{b}^{\top}\bm{\Sigma}^{-1}\bm{b}(T-t)\\
\bar{g}(t,x)=x+\frac{1}{2\lambda} \bm{b}^{\top}\bm{\Sigma}^{-1}\bm{b}(T-t).
\end{numcases}
We choose $\hat{\btheta}$ such that $\bar{V}(t, x)$ is minimized, i.e., $$\hat{\btheta}=\arg \min\limits_{\btheta\in \Theta} \bm{b}^{\top}\bm{\Sigma}^{-1}\bm{b},$$ then we have  $V(t,x)=x+\frac{1}{4\lambda}\hat{\bm{b}}^{\top}\hat{\bm{\Sigma}}^{-1}\hat{\bm{b}}(T-t)$, $g(t,x)=x+\frac{1}{2\lambda}\hat{\bm{b}}^{\top}\hat{\bm{\Sigma}}^{-1}\hat{\bm{b}}(T-t)$, and the optimal strategy $\hat{\bm{\alpha}}=\frac{1}{2\lambda}\hat{\bm{\Sigma}}^{-1}\hat{\bm{b}}$.

In the following, we verify that $(V(t, x), g(t, x))$ is {the solution to \eqref{equ ori PDE}}.
 Denote
\begin{align}	
	\mathcal{F}^{\balpha,\btheta}(V,g)=&\mathcal{L}^{\balpha,\btheta}V(t,x)-\lambda \mathcal{H}^{\balpha,\btheta}  g(t,x) \notag \\
	=&V_{t}+\balpha^{\top} \bm{b}V_{x}+\frac{1}{2}\balpha^{\top}\bm{\Sigma} \balpha V_{xx}-\lambda\balpha^{\top}\bm{\Sigma}\balpha g_{x}^{2}. \notag
\end{align}	
\textcircled{1} We show that 
\begin{align}\label{equ max alpha theta1}
	\mathcal{F}^{\balpha,\hat{\btheta}}(V,g)\leq \mathcal{F}^{\hat{\balpha},\hat{\btheta}}(V,g) = 0, \quad \forall \balpha \in \R^{n}	`.
\end{align}
To show \eqref{equ max alpha theta1}, we first have
\begin{align*}
	&\mathcal{F}^{\balpha,\hat{\btheta}}(V,g)=-\frac{1}{4\lambda}\hat{\bm{b}}^{\top}\hat{\bm{\Sigma}}^{-1}\hat{\bm{b}}+\balpha^{\top}\hat{\bm{b}}-\lambda \balpha^{\top} \hat{\bm{\Sigma}}\balpha,
\end{align*}
which is a quadratic function of $\balpha$. Then it is easy to verify \eqref{equ max alpha theta1}.
Consequently, 
\begin{align}
	\sup\limits_{\balpha \in \R^{n}} \inf\limits_{\btheta\in \Theta} \mathcal{F}^{\balpha,\btheta}(V,g)\leq 0.
\end{align}
\textcircled{2} We show that 
\begin{align}\label{equ max alpha1 theta}
	\inf_{\btheta\in \Theta}\mathcal{F}^{\hat{\balpha},\btheta}(V,g)=\mathcal{F}^{\hat{\balpha},\hat{\btheta}}(V,g)=0.
\end{align}
Denote  $H(\bm{b},\bm{\Sigma})=\bm{b}^{\top}\hat{\bm{\Sigma}}^{-1}\bm{\Sigma}\hat{\bm{\Sigma}}^{-1}\hat{\bm{b}}$, we have 
\begin{align}
	&\mathcal{F}^{\hat{\balpha},\btheta}(V,g)\notag\\
	=&-\frac{1}{4\lambda}\hat{\bm{b}}^{\top}\hat{\bm{\Sigma}}^{-1}\hat{\bm{b}}+\frac{1}{2\lambda}\hat{\bm{b}}^{\top}\hat{\bm{\Sigma}}^{-1}\bm{b}-\lambda\frac{1}{4\lambda^{2}}\hat{\bm{b}}^{\top}\hat{\bm{\Sigma}}^{-1}\bm{\Sigma}\hat{\bm{\Sigma}}^{-1}\hat{\bm{b}} \notag\\
	=&-\frac{1}{4\lambda}\hat{\bm{b}}^{\top}\hat{\bm{\Sigma}}^{-1}\hat{\bm{b}}+\frac{1}{2\lambda}H(\bm{b},\hat{\bm{\Sigma}})-\frac{1}{4\lambda}H(\hat{\bm{b}},\bm{\Sigma})\notag\\
	\geq&-\frac{1}{4\lambda}\hat{\bm{b}}^{\top}\hat{\bm{\Sigma}}^{-1}\hat{\bm{b}}+\frac{1}{4\lambda}H(\hat{\bm{b}},\hat{\bm{\Sigma}})=0.\label{equ ineq pham}
\end{align}
The last inequality \eqref{equ ineq pham} comes from Pham, Wei and Zhou (2022) \footnote{{Pham, Wei and Zhou (2022) focus on the ambiguity of correlation, i.e., volatility is known in their model, but their proof for inequality \eqref{equ zhou 2022} holds true for ambiguous covariance.   }}{, which is}
\begin{align}\label{equ zhou 2022}
H(\hat{\bm{b}},\hat{\bm{\Sigma}})-2H(\bm{b},\hat{\bm{\Sigma}})+H(\hat{\bm{b}},\bm{\Sigma}) \leq 0, \quad \forall (\bm{b}, \bm{\Sigma})\in \Theta.
\end{align}
The equality is achieved if and only if $\bm{\Sigma}=\hat{\bm{\Sigma}}, \bm{b}=\hat{\bm{b}}$. 

\textcircled{3}According to \eqref{equ max alpha theta1} and \eqref{equ max alpha1 theta}, we have for any $\balpha \in \R^{n}$, 
\begin{align}\label{equ veri final}
	\mathcal{F}^{\balpha,\hat{\btheta}}(V,g)\leq \mathcal{F}^{\hat{\balpha},\hat{\btheta}}(V,g)  =\inf_{\btheta\in \Theta}\mathcal{F}^{\hat{\balpha},\btheta} (V,g)= 0.
\end{align}
Therefore $(V, g)$ is a solution of the PDE system \eqref{equ ori PDE}.

\subsubsection{Mean-variance w.r.t. log return}
The proof is analogous to Section \ref{sec basic wealth}, we first fix $\btheta= (\bm{b}, \bm{\Sigma})$ and solve the following PDE system:
	\begin{numcases}{}
		\sup_{\balpha \in \R^{n}}\{\mathcal{L}^{\balpha,\btheta}\bar{V}(t,x)-\lambda \mathcal{H}^{\balpha,\btheta}  \bar{g}(t,x) \}=\mathcal{L}^{\hat{\balpha},\btheta} \bar{V}(t,x)-\lambda \mathcal{H}^{\hat{\balpha},\btheta}   \bar{g}(t,x)=0\notag\\
		\mathcal{L}^{\bar{\balpha},\btheta}\bar{g}(t,x)=0\notag\\
		\bar{V}(T,x)=x\notag\\
		\bar{g}(T,x)=x,\notag
	\end{numcases}
with {${\eta}(\balpha, \btheta) =\balpha^\top \bm{b} - \frac{1}{2} \balpha^\top \bm{\Sigma} \balpha$ and $\bm{\xi}(\balpha, \btheta) = (\balpha^\top \bm{\Sigma})^\top$}.
We have the following solution:
\begin{numcases}{}
	\bar{V}(t,x)=x+\frac{1}{2(1+2\lambda)}\bm{b}^{{\top}}\bm{\Sigma}^{{-1}}\bm{b}(T-t)\\
	\bar{g}(t,x)=x+{\frac{2\lambda+\frac{1}{2}}{(1+2\lambda)^{2}}}\bm{b}^{{\top}}\bm{\Sigma}^{{-1}}\bm{b}(T-t)
\end{numcases}
To minimize $\bar{V}$, we still choose $\hat{\btheta}=\arg \min\limits_{\btheta\in \Theta} \bm{b}^{\top}\bm{\Sigma}^{-1}\bm{b}$, then we have  $V(t,x)=x+\frac{1}{2(1+2\lambda)}\hat{\bm{b}}^{\top}\hat{\bm{\Sigma}}^{-1}\hat{\bm{b}}(T-t)$, $g(t,x)=x+{\frac{2\lambda+\frac{1}{2}}{(1+2\lambda)^{2}}}\hat{\bm{b}}^{\top}\hat{\bm{\Sigma}}^{-1}\hat{\bm{b}}(T-t)$, and the corresponding optimal strategy is $\hat{\balpha}=\frac{1}{1+2\lambda}\hat{\bm{\Sigma}}^{{-1}}\hat{\bm{b}}$.\\
Denote
	\begin{align}	
		\mathcal{F}^{\balpha,\btheta}(V,g)=&\mathcal{L}^{\balpha,\btheta}V(t,x)-\lambda \mathcal{H}^{\balpha,\btheta}  g(t,x) \\
		=&V_{t}+(\balpha^{\top}\bm{b}-\frac{1}{2}\balpha^{\top}\bm{\Sigma}\balpha) V_{x}+\frac{1}{2}\balpha^{\top}\bm{\Sigma} \balpha V_{xx}-\lambda\balpha^{\top}\bm{\Sigma}\balpha g_{x}^{2}.
	\end{align}	
Then we can analogously show \eqref{equ max alpha theta1} and \eqref{equ max alpha1 theta}, which implies \eqref{equ veri final}. {That implies $(V, g)$ is a solution to \eqref{equ ori PDE}}.

\subsubsection{Minimal risk premium and worst case scenario}
	Let
	\begin{align}
		f=\frac{1}{\sigma_1^{2}\sigma_2^{2}(1-\rho^2)}\bigg(\sigma_2^{2}b_1^2+\sigma_1^{2}b_2^2 -2 \sigma_1\sigma_2\rho b_1 b_2\bigg)
	\end{align}For fixed $\rho$, we consider two cases: $\rho \leq 0$ and $\rho > 0$.\\
	(1) If $\rho \leq 0$, since the first order condition is:
	\begin{align}
		&\frac{\partial f}{\partial b_1}=\frac{2}{\sigma_1^{2}(1-\rho^{2})}b_1-\frac{2}{\sigma_1\sigma_2(1-\rho^{2})}\rho b_2=\frac{2}{\sigma_1(1-\rho^{2})}(\frac{b_1}{\sigma_1}-\rho \frac{b_2}{\sigma_2}),\label{b1}\\
		&\frac{\partial f}{\partial b_2}=\frac{2}{\sigma_2^{2}(1-\rho^{2})}b_2-\frac{2}{\sigma_1\sigma_2(1-\rho^{2})}\rho b_1=\frac{2}{\sigma_2(1-\rho^{2})}(\frac{b_2}{\sigma_2}-\rho \frac{b_1}{\sigma_1}),\label{b2}\\
		&\frac{\partial f}{\partial \sigma_1}=\frac{2b_1 (\rho b_2 \sigma_1-b_1 \sigma_2)}{\sigma_1^{3}\sigma_2 (1-\rho^{2})},\label{sigma1}\\
		&\frac{\partial f}{\partial \sigma_2}=\frac{2b_2 (\rho b_1 \sigma_2-b_2 \sigma_1)}{\sigma_2^{3}\sigma_1 (1-\rho^{2})},\label{sigma2}\\
		&\frac{\partial f}{\partial \rho}=(-\frac{( b_1/\sigma_1) (b_2/\sigma_2)}{(b_1/\sigma_1)^{2}+(b_2/\sigma_2)^{2}}+\frac{\rho}{1+\rho^{2}})(1+\rho^{2})((b_1/\sigma_1)^{2}+(b_2/\sigma_2)^{2})/(1-\rho^{2})^2.\label{rho}
	\end{align}
	Notice when $\rho {\leq} 0$, $\frac{\partial f}{\partial b_1}>0$ and $\frac{\partial f}{\partial b_2}>0$, $\frac{\partial f}{\partial \sigma_1}<0$ and $\frac{\partial f}{\partial \sigma_2}<0$, then $\hat{b}_1=\underline{b}_1$, $\hat{b}_2=\underline{b}_2$, $\hat{\sigma}_1=\overline{\sigma}_1$ and $\hat{\sigma}_2=\overline{\sigma}_2$. Then we have
	
	\begin{align*}
		\frac{\partial f}{\partial \rho}=(-\frac{( \underline{b}_1/\overline{\sigma}_1) (\underline{b}_2/\overline{\sigma}_2)}{(\underline{b}_1/\overline{\sigma}_1)^{2}+(\underline{b}_2/\overline{\sigma}_2)^{2}}+\frac{\rho}{1+\rho^{2}})(1+\rho^{2})((\underline{b}_1/\overline{\sigma}_1)^{2}+(\underline{b}_2/\overline{\sigma}_2)^{2})/(1-\rho^{2})^2<0
	\end{align*}
	
	Then $\hat{\rho}=\min\{0,\overline{\rho}\}$.\\
	(2) If $\rho>0$, 
	
	{Case 1:  $\hat{b}_2/\hat{\sigma}_2\leq\hat{b}_1/\hat{\sigma}_1$.}
	
	Obviously, {according to the first order conditions,} in this case we have $\hat{b}_1=\underline{b}_1$ and $\hat{\sigma}_1=\overline{\sigma}_1$. By the first order condition, we have
	\begin{align}
		&\hat{b}_2=\overline{b}_2\quad \text{and} \quad \hat{\sigma}_2=\underline{\sigma}_2, \quad \text{if} \quad \overline{b}_2/\underline{\sigma}_2<\rho \underline{b}_1/\overline{\sigma}_1,\label{1}\\ 
		&{\frac{\hat{b}_2}{\hat{\sigma}_2} = \hat{\rho} \frac{\underline{b}_1}{\bar{\sigma}_1}}, \quad \text{if} \quad \underline{b}_2/\overline{\sigma}_2<\rho \underline{b}_1/\overline{\sigma}_1<\overline{b}_2/\underline{\sigma}_2,\label{2}\\
		&\hat{b}_2=\underline{b}_2\quad \text{and} \quad \hat{\sigma}_2=\overline{\sigma}_2, \quad \text{if} \quad \underline{b}_2/\overline{\sigma}_2>\rho \underline{b}_1/\overline{\sigma}_1.\label{3}
	\end{align}
	In case (\ref{1}), $1>\rho>\frac{\overline{b}_2/\underline{\sigma}_2}{\underline{b}_1/\overline{\sigma}_1}$. By (\ref{rho}), we have $\hat{\rho}=\max \{\underline{\rho},\frac{\overline{b}_2/\underline{\sigma}_2}{\underline{b}_1/\overline{\sigma}_1}\}$.\\
	In case (\ref{2}), {$\hat{b}_2, \hat{\sigma}_2, \hat{\rho}$ is chosen to satisfy (\ref{2})}. Notice, in this case, $f=\underline{b}_1^{2}/\overline{\sigma}_1^{2}$.\\
	In case (\ref{3}), $\rho<\frac{\underline{b}_2/\overline{\sigma}_2}{\underline{b}_1/\overline{\sigma}_1}$. By (\ref{rho}), we obtain $\hat{\rho}=\min \{\overline{\rho},\frac{\underline{b}_2/\overline{\sigma}_2}{\underline{b}_1/\overline{\sigma}_1}\}$.

	{Case 2: $\hat{b}_2/\hat{\sigma}_2>\hat{b}_1/\hat{\sigma}_1$}
	
{	Analogous to Case 1, we have $\hat{b}_2=\underline{b}_2$ and $\hat{\sigma}_2=\overline{\sigma}_2$ from the first order conditions. Then 
	\begin{align}
		&\hat{b}_1=\overline{b}_1\quad \text{and} \quad \hat{\sigma}_1=\underline{\sigma}_1, \quad \text{if} \quad \overline{b}_1/\underline{\sigma}_1<\rho \underline{b}_2/\overline{\sigma}_2,\label{equ worst case 221}\\ 
		&{\frac{\hat{b}_1}{\hat{\sigma}_1} = \hat{\rho} \frac{\underline{b}_2}{\bar{\sigma}_2}}, \quad \text{if} \quad \underline{b}_1/\overline{\sigma}_1<\rho \underline{b}_2/\overline{\sigma}_2<\overline{b}_1/\underline{\sigma}_1,\label{equ worst case 222}\\
		&\hat{b}_1=\underline{b}_1\quad \text{and} \quad \hat{\sigma}_1=\overline{\sigma}_1, \quad \text{if} \quad \underline{b}_1/\overline{\sigma}_1>\rho \underline{b}_2/\overline{\sigma}_2.\label{equ worst case 223}
	\end{align}
  Notice that \eqref{equ worst case 221} and \eqref{equ worst case 222} contradict the Assumption \eqref{assump b sigma}, while \eqref{equ worst case 223} contradicts  the assumption $\hat{b}_2/\hat{\sigma}_2>\hat{b}_1/\hat{\sigma}_1$. Therefore, this Case 2 never happens.
}	
	
	Combining {the above results}, we obtain Proposition 2.

\subsection{Explicit solutions in Section \ref{sec expli exten}}
\subsubsection{Jump with compounded Possion process and constant $\lambda$}
Similar as Section \ref{sec basic wealth},
we first fix $\btheta = (\bm{b}, \bm{\Sigma}, \bm{\mu}, \bm{J}, \bm{\Phi})$, and solve the following PDE system. 
\begin{align*}
	\left\{\begin{matrix}
		\sup\limits_{\balpha\in \R^{n}}\{ \mathcal{A}^{\balpha,\btheta}\bar{f}^{x}(t,x)+2\lambda \bar{g} \mathcal{A}^{\balpha,\btheta}\bar{g}(t,x)\}=0\\
		\mathcal{A}^{\hat{\balpha}, \hat{\btheta}}\bar{f}^x(t,x)=0\\
		\mathcal{A}^{\hat{\balpha}, \hat{\btheta}}\bar{g}(t,x)=0\\
		\bar{f}^x(T,x)=x-\lambda x^{2}\\
		\bar{g}(T,x)=x,
	\end{matrix}\right.
\end{align*}
where 
 $\mathcal{A}^{\balpha, \btheta}\psi(t,x)=\psi_{t}(t,x)+\balpha^{\top}\bm{b^{\theta}}\psi_{x}(t,x)+\frac{1}{2}\balpha^{\top}\bm{\sigma} \balpha\psi_{xx}(t,x)+\sum_{l=1}^{k}\mu_{l} \mathbb{E}[\psi(t,x+\balpha^{\top}\bm{J}_{l} Y_{l})-\psi(t,x)]$.

We have a solution $(\bar{f}^x, \bar{g})$ as follows: 
\begin{numcases}{}
\bar{f}^x(t,x)=-\lambda x^{2}+\bigg(1-{\bm{b}}_{\bm{F}}^\top{\bm{\Sigma}}_{\bm{F}}^{-1}{\bm{b}}_{\bm{F}}(T-t)\bigg)x +\frac{1}{4\lambda}{\bm{b}}_{\bm{F}}^\top{\bm{\Sigma}}_{\bm{F}}^{-1}{\bm{b}}_{\bm{F}}(T-t)\bigg(1-{\bm{b}}_{\bm{F}}^\top{\bm{\Sigma}}_{\bm{F}}^{-1}{\bm{b}}_{\bm{F}}(T-t)\bigg)\notag\\
    \bar{g}(t,x)=x+\frac{1}{2\lambda}{\bm{b}}_{\bm{F}}^\top{\bm{\Sigma}}_{\bm{F}}^{-1}{\bm{b}}_{\bm{F}}(T-t).\notag
\end{numcases}
Assume $\hat{\btheta}=(\bm{\hat{b}_{F}},\hat{\bm{\Sigma}}_{\F})=\arg\inf\limits_{\btheta \in \Theta}{\bm{b}}_{\bm{F}}^\top{\bm{\Sigma}}_{\bm{F}}^{-1}{\bm{b}}_{\bm{F}}$, then we have
\begin{numcases}{}
 f^x(t,x)=-\lambda x^{2}+\bigg(1-\hat{\bm{b}}_{\bm{F}}^\top\hat{\bm{\Sigma}}_{\bm{F}}^{-1}\hat{\bm{b}}_{\bm{F}}(T-t)\bigg)x +\frac{1}{4\lambda}\hat{\bm{b}}_{\bm{F}}^\top\hat{\bm{\Sigma}}_{\bm{F}}^{-1}\hat{\bm{b}}_{\bm{F}}(T-t)\bigg(1-\hat{\bm{b}}_{\bm{F}}^\top\hat{\bm{\Sigma}}_{\bm{F}}^{-1}\hat{\bm{b}}_{\bm{F}}(T-t)\bigg)\notag\\ g(t,x)=x+\frac{1}{2\lambda}\hat{\bm{b}}_{\bm{F}}^\top\hat{\bm{\Sigma}}_{\bm{F}}^{-1}\hat{\bm{b}}_{\bm{F}}(T-t). \notag
 \end{numcases}
 Similar as Section \ref{sec basic wealth}, we can verify
\begin{align}
	\mathcal{F}^{\balpha,\hat{\btheta}}(V,g)\leq \mathcal{F}^{\hat{\balpha},\hat{\btheta}}(V,g)  =\inf_{\btheta\in \Theta}\mathcal{F}^{\hat{\balpha},\btheta} (V,g)= 0, \quad {\forall \balpha\in \R^{n}},
\end{align}
where {$V(t, x) = f^x(t,x)+\lambda(x)g^{2}(t,x)$},
That implies $(f^x, g)$ is a solution. 

\subsubsection{Jump and $\lambda(x) = \frac{\lambda}{x}$}
First,	we fix $\btheta = (\bm{b}, \bm{\Sigma}, \bm{F})$ and find a solution of the PDE system:	
	\begin{align}\label{equ3 ori PDE}
		\left\{\begin{matrix}
			\sup\limits_{\balpha\in \R^{n}}\{ \mathcal{A}^{\balpha,\btheta}\bar{f}^{x}(t,x)+2\frac{\lambda}{x}\bar{g} \mathcal{A}^{\balpha,\btheta}\bar{g}(t,x)\}=0\\
			\mathcal{A}^{\hat{\balpha}, \btheta}\bar{f}^x(t,x)=0\\
			\mathcal{A}^{\hat{\balpha},\btheta}\bar{g}(t,x)=0\\
			\bar{f}^x(T,x)=x-{\lambda}x\\
			\bar{g}(T,x)=x,
		\end{matrix}\right.
	\end{align}
	Similary, according to Bjork (2014), we have the following result that for all $\btheta$, 
	\begin{numcases}{}
		\hat{\balpha}=\frac{1}{2\lambda}\bm{\Sigma}_{\F}^{{-1}}\bm{b}_{\F}\frac{\bar{A}(t)+2\lambda [\bar{A}^{2}(t)-\bar{B}(t)]}{\bar{B}(t)}x\notag\\
		\bar{f}^x(t,x)=\bar{A}(t)x-{\lambda}\bar{B}(t)x\notag\\
		\bar{g}(t,x)=\bar{A}(t)x \notag
	\end{numcases}
	where $\bar{A}$ and $\bar{B}$ solve the ODE system
	\begin{numcases}{}
		\bar{A}_{t}+(\frac{\bar{A}+2\lambda[\bar{A}^{2}-\bar{B}]}{2\lambda \bar{B}}{\bm{b}}_{\bm{F}}^\top{\bm{\Sigma}}_{\bm{F}}^{-1}{\bm{b}}_{\bm{F}})\bar{A}=0\label{ode1}\label{equ AB ODE system1}\\
		\bar{A}(T)=1\\
		\bar{B}_{t}+\{2(\frac{\bar{A}+2\lambda[\bar{A}^{2}-\bar{B}]}{2\lambda \bar{B}}{\bm{b}}_{\bm{F}}^\top{\bm{\Sigma}}_{\bm{F}}^{-1}{\bm{b}}_{\bm{F}})+\frac{(\bar{A}+2\lambda[\bar{A}^{2}-\bar{B}])^{2}}{4\lambda^{2} \bar{B}^{2}}{\bm{b}}_{\bm{F}}^\top{\bm{\Sigma}}_{\bm{F}}^{-1}{\bm{b}}_{\bm{F}}\}\bar{B}=0\label{ode2}\\
		\bar{B}(T)=1\label{equ AB ODE system4}.
	\end{numcases}
	Then we guess ${\hat{\btheta}=(\hat{\bm{b}},\hat{\bm{\Sigma}}, \hat{\bm{F}})}=\arg\inf\limits_{\btheta \in \Theta}\bm{b}^{{\top}}_{\F}\bm{\Sigma}^{{-1}}_{\F}\bm{b}_{\F}$, i.e.,
	\begin{numcases}{}
		\hat{\balpha}=\frac{1}{2\lambda}\hat{\bm{\Sigma}}_{\F}^{-1}\hat{\bm{b}}_{\F}\frac{A(t)+2\lambda [A^{2}(t)-B(t)]}{B(t)}x\notag\\
		f^x(t,x)=A(t)x-\frac{\lambda}{y}B(t)x^{2}\notag\\
		g(t,x)=A(t)x\notag\\
		\hat{\btheta}=\arg \inf_{\btheta\in \Theta} {\bm{b}}_{\bm{F}}^\top{\bm{\Sigma}}_{\bm{F}}^{-1}{\bm{b}}_{\bm{F}},\notag
	\end{numcases}
{where $(A(t), B(t))$ is the solution $(\bar{A}(t), \bar{B}(t))$ to \eqref{equ AB ODE system1}-\eqref{equ AB ODE system4} with $\btheta = \hat{\btheta}$.}
	
We next show that our guess is indeed a solution to \eqref{equ2 ori PDE}. 
{Similar as Section \ref{sec basic wealth}}, if we can show that $\sup \limits_{\balpha\in \R^{n}}{\mathcal{F}}^{\balpha, \hat{\btheta}}(V,g)=0$ and $\inf\limits_{\btheta\in \Theta}{\mathcal{F}}^{\hat{\balpha}, \btheta}(V,g) =0$ {for $V(t,x) = f^x(t, x)+ \frac{\lambda}{x} g^2(t, x)$}, {then we have verified our guess}.\\
\textcircled{1}{We first show} $\sup \limits_{\balpha\in \R^{n}}{\mathcal{F}}^{\balpha, \hat{\btheta}}(V,g)=0$.
	
We have
	\begin{align*}
		{\mathcal{F}}^{\balpha,\hat{\btheta}}(V,g)=A_{t}x-\lambda B_{t}x+2\lambda AA_{t}x+\balpha^{\top}\hat{\bm{b}}_{\F}(A-2\lambda B+2\lambda A^{2})-\frac{1}{2}\balpha^{\top}\hat{\bm{\Sigma}}_{\F}\balpha(\frac{2\lambda}{x}B).
	\end{align*}
	Since ${\mathcal{F}}^{\balpha,\hat{\btheta}}(V,g)$ is a quadratic function of $\balpha$, it is easy to verify \textcircled{1}.
	Consequently, 
	\begin{align}
		\sup_{\balpha\in \R^{n}} \inf\limits_{\btheta\in \Theta} \mathcal{F}^{\balpha,\btheta}(V,g)\leq \sup \limits_{\balpha\in \R^{n}}{\mathcal{F}}^{\balpha, \hat{\btheta}}(V,g)\leq0.
	\end{align}
	
\textcircled{2} {We next show} $\inf\limits_{\btheta\in \Theta}{\mathcal{F}}^{\hat{\balpha}, \btheta}(V, g)=0$.
	\begin{align*}
		&{\mathcal{F}}^{\hat{\balpha},\btheta}(V,g)=A_{t}x-\lambda B_{t}x+2\lambda AA_{t}x+\hat{\balpha}^{{\top}}\bm{b}_{\F}(A-2\lambda B+2\lambda A^{2})-\frac{1}{2}\hat{\balpha}^{{\top}}\bm{\Sigma}_{\F}\hat{\balpha}(\frac{2\lambda}{x}B)\\
		&=A_{t}x-\lambda B_{t}x+2\lambda AA_{t}x+\frac{(A-2\lambda B+2\lambda A^{2})^{2}x}{2\lambda B}H(\bm{b}_{\F},\hat{\bm{\Sigma}}_{\F})-\frac{1}{2}\frac{(A-2\lambda B+2\lambda A^{2})^{2}x}{2\lambda B}H(\hat{\bm{b}}_{\F},\bm{\Sigma}_{\F})\\
		&\geq A_{t}x-\lambda B_{t}x+2\lambda AA_{t}x+\frac{(A-2\lambda B+2\lambda A^{2})^{2}x}{4\lambda B}\hat{\bm{b}}^{\top}_{\F}\hat{\bm{\Sigma}}_{\F}^{-1}\hat{\bm{b}}_{\F}
	\end{align*}
	{where the last inequality is from 
	\begin{align}
		H(\hat{\bm{b}}_{\F},\bm{\Sigma}_{\F}) \leq H(\hat{\bm{b}}_{\F},\hat{\bm{\Sigma}}_{\F}) =\hat{\bm{b}}^{\top}_{\F}\hat{\bm{\Sigma}}_{\F}^{-1}\hat{\bm{b}}_{\F} \leq H(\bm{b}_{\F},\hat{\bm{\Sigma}}_{\F})
	\end{align}
    as \eqref{equ zhou 2022}.
    }
	Let $\eqref{ode1}\times(x+2\lambda Ax)-\eqref{ode2}\times\lambda x$, we have
	\begin{align*}
		&\eqref{ode1}\times(x+2\lambda Ax)-\eqref{ode2}\times\lambda x\\
		&=A_{t}x-\lambda B_{t}x+2\lambda AA_{t}x+\frac{(A-2\lambda B+2\lambda A^{2})^{2}x}{4\lambda B}\hat{\bm{b}}^{\top}_{\F}\hat{\bm{\Sigma}}_{\F}^{-1}\hat{\bm{b}}_{\F}=0
	\end{align*}
	And the equality is achieved if and only if when $\bm{\Sigma}_{\F}=\hat{\bm{\Sigma}}_{\F}, \bm{b}_{\F}=\hat{\bm{b}}_{\F}$.

\end{document}